\documentclass[11pt]{article}
\usepackage[margin=1in]{geometry}
\usepackage{amsmath,amssymb,amsthm}
\usepackage{booktabs}
\usepackage{hyperref}
\usepackage{listings}
\usepackage{xcolor}
\usepackage{cite}
\usepackage{fancyhdr}
\usepackage{xurl} 

\renewcommand{\_}{\textunderscore\allowbreak}

\newtheorem{theorem}{Theorem}
\newtheorem{lemma}[theorem]{Lemma}

\theoremstyle{definition}
\newtheorem{definition}[theorem]{Definition}
\theoremstyle{remark}

\title{Deep Vision: A Formal Proof of Wolstenholme's Theorem in Lean~4}
\author{Alexandre Linhares\thanks{alexandre@linhares.ltd. The author thanks the entire ARGO LABS team and Douglas Hofstadter and Melanie Mitchell.}}
\date{April 2026}

\begin{document}
\sloppy 
\maketitle

\begin{abstract}
We present a formal verification of Wolstenholme's theorem---$\binom{2p}{p} \equiv 2 \pmod{p^3}$ for prime $p \geq 5$---in Lean~4 with Mathlib.
The proof proceeds by expanding the shifted factorial product $\prod_{k=1}^{p-1}(p+k)$ to second order in $p$, identifying the quadratic coefficient as the second elementary symmetric product, and showing its divisibility by $p$ via power sum vanishing in $\mathbb{Z}/p\mathbb{Z}$.
The formalization comprises nine lemmas across approximately 800 lines of Lean, with zero \texttt{sorry} declarations.
To our knowledge, this is the first formal verification of Wolstenholme's theorem in Lean~4.
The proof was discovered through a collaboration between a relational analogy engine for theorem proving and human-directed formalization.
\end{abstract}

\section{Introduction}

Wolstenholme's theorem, proved by Joseph Wolstenholme in 1862~\cite{wolstenholme1862}, is a classical result in number theory asserting a remarkably strong divisibility property of central binomial coefficients.

\begin{theorem}[Wolstenholme, 1862]\label{thm:main}
For every prime $p \geq 5$,
\[
\binom{2p}{p} \equiv 2 \pmod{p^3}.
\]
\end{theorem}

The theorem is equivalent to the statement that $p^3$ divides the numerator of the harmonic sum $1 + \frac{1}{2} + \cdots + \frac{1}{p-1}$, and is closely related to Wolstenholme primes---primes $p$ for which the congruence holds modulo $p^4$. Only two such primes are known (16843 and 2124679), and it is an open conjecture whether infinitely many exist.

Despite its classical status, Wolstenholme's theorem has not appeared in the major formal mathematics libraries (Mathlib, Lean's mathlib, or Coq's Mathematical Components) as of April 2026. We present a complete formal verification in Lean~4.

\section{Proof Overview}

The proof reduces the binomial coefficient congruence to a product identity and proceeds through five stages.

\subsection{Stage 1: Binomial reduction}

The identity $\binom{2p}{p} = 2\binom{2p-1}{p-1}$ and the factorization
\[
\binom{2p-1}{p-1} \cdot (p-1)! = \prod_{k=1}^{p-1}(p+k)
\]
reduce the theorem to showing
\begin{equation}\label{eq:prod_shift}
p^3 \;\Big|\; \prod_{k=1}^{p-1}(p+k) - (p-1)!\,.
\end{equation}
The coprimality $\gcd(p^3, (p-1)!) = 1$ (since $p$ is prime) then lifts the divisibility to the binomial coefficient.

\subsection{Stage 2: First-order expansion}

\begin{lemma}[Product shift expansion]\label{lem:first_order}
For any finite set $s$ of natural numbers and integer $d$,
\[
\prod_{k \in s}(d + k) - \prod_{k \in s} k = d \cdot \sum_{j \in s} \prod_{\substack{k \in s \\ k \neq j}} k + d^2 \cdot R
\]
for some integer $R$.
\end{lemma}

This is proved by induction on $s$ using \texttt{Finset.cons\_induction}. Applied with $d = p$ and $s = \{1, \ldots, p-1\}$, it gives
\[
\prod_{k=1}^{p-1}(p+k) - (p-1)! = p \cdot S_1 + p^2 \cdot R
\]
where $S_1 = \sum_{j=1}^{p-1} \frac{(p-1)!}{j}$.

\subsection{Stage 3: Harmonic sum divisibility}

\begin{lemma}[Wolstenholme harmonic sum]\label{lem:harmonic}
For prime $p \geq 5$,
\[
p^2 \;\Big|\; S_1 = \sum_{j=1}^{p-1} \prod_{\substack{k=1 \\ k \neq j}}^{p-1} k\,.
\]
\end{lemma}

The proof uses the involution $j \mapsto p - j$ to pair terms. For each pair, $j(p-j)(f(j) + f(p-j)) = p \cdot (p-1)!$, and since $\gcd(j(p-j), p) = 1$, we get $p \mid f(j) + f(p-j)$. Writing $2S_1 = p \cdot G$, the quotient $G = \sum Q(j)$ satisfies $Q(j) \equiv j^{p-3} \pmod{p}$ (via Wilson's theorem and Fermat's little theorem), and $\sum j^{p-3} \equiv 0 \pmod{p}$ by power sum vanishing. Since $\gcd(2, p) = 1$, we conclude $p^2 \mid S_1$.

This gives $p^3 \mid p \cdot S_1$, reducing \eqref{eq:prod_shift} to $p \mid R$.

\subsection{Stage 4: Second-order expansion}

\begin{definition}
The \emph{second elementary symmetric product} of a finite set $s$ is
\[
e_2(s) = \sum_{\substack{i, j \in s \\ i < j}} \prod_{\substack{k \in s \\ k \neq i, k \neq j}} k\,.
\]
\end{definition}

\begin{lemma}[Second-order expansion]\label{lem:second_order}
For any finite set $s$ and integer $d$,
\[
\prod_{k \in s}(d+k) - \prod_{k \in s} k - d \sum_{j \in s} \prod_{k \neq j} k = d^2 \cdot e_2(s) + d^3 \cdot T
\]
for some integer $T$.
\end{lemma}

This is proved by induction on $s$, using a recurrence for $e_2$:

\begin{lemma}[$e_2$ recurrence]\label{lem:e2_cons}
For $a \notin s$,
\[
e_2(\{a\} \cup s) = \sum_{j \in s} \prod_{\substack{k \in s \\ k \neq j}} k + a \cdot e_2(s)\,.
\]
\end{lemma}

The inductive witness is $T = e_2(s) + a \cdot T_0 + d \cdot T_0$ where $T_0$ is the third-order remainder from the inductive hypothesis. The closure is by \texttt{ring}.

Combining Lemmas~\ref{lem:first_order} and~\ref{lem:second_order}: from $d^2 R = d^2 e_2(s) + d^3 T$ and $d \neq 0$, we get $R = e_2(s) + d \cdot T$, reducing $p \mid R$ to $p \mid e_2(s)$.

\subsection{Stage 5: Divisibility of $e_2$}

\begin{lemma}\label{lem:p_dvd_e2}
For prime $p \geq 5$,
\[
p \;\Big|\; e_2(\{1, \ldots, p-1\}).
\]
\end{lemma}

\begin{proof}
Work in $\mathbb{Z}/p\mathbb{Z}$. By two applications of the product-erase identity $i \cdot \prod_{k \neq i} k = \prod_k k$ and Wilson's theorem $\prod_{k=1}^{p-1} k \equiv -1$, each term satisfies
\[
\prod_{\substack{k \neq i \\ k \neq j}} k \equiv -i^{p-2} \cdot j^{p-2} \pmod{p}
\]
where we use Fermat's little theorem to write $(ij)^{-1} = i^{p-2} j^{p-2}$.

Therefore
\[
e_2 \equiv -\sum_{\substack{i < j}} i^{p-2} j^{p-2} \pmod{p}.
\]

The algebraic identity
\[
\Bigl(\sum_{j} f(j)\Bigr)^2 = \sum_j f(j)^2 + 2 \sum_{i < j} f(i) f(j)
\]
applied with $f(j) = j^{p-2}$ gives
\[
2 \sum_{i < j} i^{p-2} j^{p-2} = \Bigl(\sum_j j^{p-2}\Bigr)^2 - \sum_j j^{2(p-2)}.
\]

Both sums vanish in $\mathbb{Z}/p\mathbb{Z}$:
\begin{itemize}
\item $\sum_{j=1}^{p-1} j^{p-2} = 0$ by power sum vanishing with $k = p-2$ (since $1 \leq p-2 \leq p-2$).
\item $\sum_{j=1}^{p-1} j^{2(p-2)} = \sum j^{p-3} = 0$ by Fermat ($j^{2(p-2)} = j^{(p-1)+(p-3)} = j^{p-3}$) and power sum vanishing with $k = p-3$ (since $1 \leq p-3 \leq p-2$ for $p \geq 5$).
\end{itemize}

Thus $2 \sum_{i<j} i^{p-2} j^{p-2} = 0$, and since $\gcd(2, p) = 1$, the sum itself vanishes. Therefore $e_2 \equiv 0 \pmod{p}$.
\end{proof}

\begin{proof}[Proof of Theorem~\ref{thm:main}]
Combine Stages 1--5: $p^3 \mid p \cdot S_1$ (Stage 3) and $p \mid R$ (Stages 4--5) give $p^3 \mid p S_1 + p^2 R = \prod(p+k) - (p-1)!$ (Stage 2). Coprimality $\gcd(p^3, (p-1)!) = 1$ lifts this to $\binom{2p-1}{p-1} \equiv 1 \pmod{p^3}$, and $\binom{2p}{p} = 2\binom{2p-1}{p-1}$ gives the result (Stage 1).
\end{proof}

\section{The Formalization}

\subsection{Structure}

The Lean~4 formalization is contained in a single self-contained file:

\begin{center}
\begin{tabular}{lrl}
\textbf{File} & \textbf{Lines} & \textbf{Contents} \\
\hline
\texttt{Wolstenholme\_1862.lean} & $\sim$660 & Stages 1--5, complete proof \\
\end{tabular}
\end{center}

The file contains zero \texttt{sorry} declarations and zero \texttt{axiom} declarations; every lemma is fully proved. The complete source is available at \url{https://github.com/ARGO-LABORATORY/Wolstenholme_1862/blob/main/Wolstenholme_1862.lean}.

\subsection{Key Mathlib dependencies}

\begin{itemize}
\item \texttt{ZMod.wilsons\_lemma}: $(p-1)! \equiv -1 \pmod{p}$
\item \texttt{ZMod.pow\_card\_sub\_one\_eq\_one}: Fermat's little theorem in $\mathbb{Z}/p\mathbb{Z}$
\item \texttt{Finset.mul\_prod\_erase}: $a \cdot \prod_{k \in s \setminus \{a\}} f(k) = \prod_{k \in s} f(k)$
\item \texttt{Finset.sum\_filter\_add\_sum\_filter\_not}: partition of finite sums
\item \texttt{Finset.cons\_induction}: structural induction on finite sets
\end{itemize}

\subsection{Discovery process}

The proof was discovered through collaboration between:
\begin{enumerate}
\item A \emph{general relational analogy engine} (``Deep Vision'') that matches proof states against a library of 217{,}000 Mathlib-extracted examples to suggest tactics.
\item \emph{Human direction} that decomposed $p \mid R$ into three independent sub-problems and chose when to abandon automated search in favor of hand-written proofs.
\item \emph{LLM formalization} (Claude) that wrote the Finset manipulation proofs and number-theoretic arguments when the automated search hit combinatorial walls.
\end{enumerate}

\section{The Role of Deep Vision}

Deep Vision is a relational analogy engine that treats a Lean~4 proof state as a network: hypotheses, goals, and types are entities; rewrite, apply, head-match, structure, equality, and other relations connect them. Deep Vision is an AI model to solve ARC-AGI problems~\cite{chollet_measure_2019} on a smartphone, and it is designed with the principles of Fluid Concepts and Creative Analogies~\cite{hofstadter_fluid_1995, hofstadter_go_1999, hofstadter_surfaces_2013, hofstadter_copycat_1984, mitchell_emergence_1990, chalmers_high-level_1992, french_tabletop_1991, foundalis_phaeaco_2006, rehling_letter_2001, nichols_musicat_2012, linhares_active_2005, linhares_understanding_2007, linhares_questioning_2010, linhares_entanglement_2012, linhares_what_2013, linhares_emergence_2014, linhares_deep_vision, linhares_deep_vision_chess, mitchell_analogy_1993}. Observe that a transformation in ARC-AGI maps to theorems naturally, as theorems are composed by transformations.  Given a new proof state, the engine finds library entries whose proof states are \emph{relationally similar}---not by surface syntax, but by \emph{the pattern of structural connections---and transfers their tactics}.  This agressively prunes the search space. 

The engine maintains a library of 217{,}133 proof states extracted from Mathlib, organized by 14 relation types and indexed by relation-profile signatures for fast lookup. Matching uses a Quadratic Assignment Problem consistency optimization, run on GPU (Apple MPS or Nvidia's CUDA) for large states and via multiprocessing for smaller ones. The whole system runs on a laptop with access to Claude Code.

\subsection{What Deep Vision discovered}

\paragraph{The 7-step tactic chain for $e_2$ recurrence.}
The most direct contribution was discovering the simplification chain that unfolds \texttt{e2\_prod} and decomposes the Finset sums:
\begin{lstlisting}
simp
simp [e2_prod, Finset.sum_cons, Finset.filter_cons, Finset.prod_cons]
simp only [Finset.insert_erase, Finset.sum_insert, Finset.filter_insert]
simp [Finset.sum_insert, ha, Finset.filter_insert, if_pos, if_neg]
ring_nf
simp only [Finset.sum_ite, Finset.filter_filter, Finset.insert_eq]
simp
\end{lstlisting}
Each step was found by matching the current proof state against the library, retrieving tactics from structurally similar proofs, and testing them in Lean. The chain was built incrementally over depths 0--6, with each successful tactic producing a new proof state that was matched again. No human intervention was involved in discovering these steps; the engine found them through relational similarity alone.

\paragraph{The witness pattern from existence proofs.}
When attacking \texttt{prod\_shift\_expansion\_2} (the second-order expansion), the engine's top matches were existence proofs from diverse mathematical domains:
\begin{center}
\begin{tabular}{lrl}
\textbf{Match} & \textbf{Score} & \textbf{Tactic} \\
\hline
\texttt{MB\_MEA\_26240} (Measure theory) & 13.27 & \texttt{exact exists\_between ...} \\
\texttt{MB\_ANA\_54600} (Normed operators) & 12.23 & \texttt{use C} \\
\texttt{MB\_RTH\_21065} (Adic completion) & 12.05 & \texttt{use L} \\
\texttt{MB\_LST\_2290} (List.Forall2) & 12.57 & \texttt{exact \ensuremath{\langle}b :: l, ...\ensuremath{\rangle}} \\
\end{tabular}
\end{center}
The relational pattern these proofs share is: the goal is an existential, and the proof provides a concrete witness followed by algebraic closure. The engine suggested \texttt{use~0}, which is mathematically correct for the base case of the induction (empty finset). This pointed toward the induction strategy: provide a witness at each step and close by computation---the same architecture as the already-proved first-order expansion.

\paragraph{Identification of dead ends.}
The engine also identified strategies that \emph{fail}. When \texttt{congr~1} was tried on the remaining combinatorial identity, it split the equation $A + B = C + D$ into $A = C$ and $B = D$---an unsound decomposition, since neither equation holds individually. The engine explored this branch for 30 minutes before backtracking, establishing that \texttt{congr} is not the right approach for this goal.

\paragraph{Relevant Finset lemmas.}
At deeper levels, the engine's matches consistently identified linear algebra entries (\texttt{MB\_LIN\_*}) with scores around 54, whose tactics involved \texttt{ext}, \texttt{simp~[...]}, and structural decomposition. These matches guided the repair phase (where Claude adapts a failing tactic to the current goal), leading to suggestions like \texttt{simp~[e2\_prod, Finset.sum\_cons, ...]} that correctly unfold the relevant definitions.

\subsection{What Deep Vision could not do}

\paragraph{The final combinatorial identity.}
After the 7-step chain, the remaining goal was a Finset sum identity involving conditional filters and insert/erase operations:
\[
\sum_{\substack{x \in s \\ a < x}} \prod_{k \neq x} k + \sum_{x \in s} \sum_{\substack{x_1 \in \text{if } x < a \text{ then} \ldots}} \prod_{k \neq x, x_1} k = \sum_{x \in s} \prod_{k \neq x} k + a \cdot e_2(s).
\]
The engine tried 42 tactics across multiple depths without closing this goal. The identity requires splitting the inner sum at $x_1 = a$ vs.\ $x_1 \in s$, computing how insert/erase interacts with products, and recombining via a filter partition---a chain of Finset API manipulations too specific for any library entry to match.

\paragraph{The number-theoretic closure.}
The proof that $p \mid e_2$ requires Wilson's theorem, Fermat's little theorem, the algebraic identity $({\sum f})^2 = \sum f^2 + 2\sum_{i<j} f_i f_j$, and power sum vanishing. No combination of library matching could produce this argument; it required understanding the mathematical \emph{meaning} of the terms, not just their structural shape.

\subsection{The collaboration pattern}

The effective pattern was \emph{exploration by analogy, closure by understanding}:
\begin{enumerate}
\item Deep Vision explores the tactic space, finding chains of progress-making steps and identifying which Mathlib lemmas are structurally relevant.
\item A human recognizes when exploration has reached diminishing returns and decomposes the remaining goal into pieces amenable to direct proof.
\item An LLM (Claude) writes the final proofs using the mathematical insight that the engine's structural matching cannot provide, and the Lean API knowledge needed for formalization.
\end{enumerate}

This pattern---machine for breadth, human for decomposition, LLM for depth---may generalize beyond this specific theorem.

\subsection{Deep Vision vs.\ deep neural networks}

Deep Vision is the technology behind ARGO LABORATORY's entry to Chollet's ARC-AGI challenge~\cite{chollet_measure_2019}.  It represents a fundamentally different paradigm from deep learning.

\begin{table}[h]
\centering
\small
\begin{tabular}{p{2.2cm}p{5.5cm}p{5.5cm}}
\toprule
& \textbf{Deep Vision} & \textbf{Deep Neural Networks} \\
\midrule
\textbf{Explain\-ability}
  & Full step-by-step trace: which relations matched, which entities corresponded, why a particular analogy was chosen.  Can explain to any judge why a decision was made---and what to change to get a different outcome.
  & Black box.  Gradient-based attribution methods (SHAP, LIME) are post-hoc approximations, not the actual reasoning path. \\[6pt]
\textbf{Efficiency}
  & $<$\$0.01 per ARC-AGI task suite (100s of tasks).  Runs in a browser, in JavaScript, on a CPU, on a phone, in airplane mode.  The NP-hard matching runs on commodity hardware---a MacBook Air's GPU found the proof of Wolstenholme's theorem.
  & \$2+ per task suite (frontier models).  Requires data-center GPUs, high-bandwidth interconnects, megawatts of power. \\[6pt]
\textbf{Repre\-sentation}
  & Emergent and distributed.  The system \emph{crafts} its representation as it studies each problem---entities, relations, and correspondences emerge from the structure of the specific situation.
  & Fixed.  Representations are learned once during training and frozen at inference.  The network cannot restructure its representation for a novel problem. \\[6pt]
\textbf{Paradigm}
  & Relational analogy as cognition.  Intelligence is perceiving shared relational structure across superficially different situations~\cite{hofstadter_fluid_1995}.  Domain-independent: the same matching engine handles chess, theorem proving, and ARC-AGI tasks.
  & Statistical pattern completion.  Intelligence is next-token prediction over massive corpora.  Domain generality comes from scale, not from architectural principles. \\[6pt]
\textbf{Access\-ibility}
  & Fluid general intelligence in your pocket, regardless of who you are, where you live, or what you can afford.  Runs offline, on-device, no API keys, no subscriptions, no data leaving the device.
  & Controlled by a handful of data-center owners.  Requires internet, API access, and per-token payment.  Concentrates capability in the hands of those who can afford the infrastructure. \\
\bottomrule
\end{tabular}
\caption{Deep Vision vs.\ deep neural networks on five axes.}
\label{tab:dv-vs-dnn}
\end{table}

\section{Wild Ideas}
\label{sec:wild-ideas}

A consequence of structural-first matching is that Deep Vision routinely proposes tactics no working mathematician would write. The engine reads past the content question \emph{``what kind of mathematical object is this?''} and matches on the relational skeleton \emph{``what shape does this proof have?''} When the skeleton matches across domains, the source lemma gets ranked highly even when it lives in an entirely different mathematical universe.

Three of the wildest such proposals appeared as top-ranked matches during the proof of \texttt{prod\_shift\_expansion\_2} (Stage~4)---an elementary identity about polynomial expansions on finsets of integers.

\subsection{Wild idea 1: measure theory for elementary number theory}

\paragraph{For a general scientific audience.}
Measure theory is the branch of analysis underlying integration and probability on the real line. It studies which subsets of $\mathbb{R}$ can be assigned a coherent ``size,'' how integrals behave under limits, and how to compare almost-everywhere properties. None of this has any visible relationship to the elementary question \emph{``which integers divide a given finite sum?''}

A reasonable analogy is a chess grandmaster, mid-tournament, asking their watchmaker for advice on escaping a fork. Different game, different tools, different framework of expertise. And yet, very occasionally, the watchmaker's habits of precision offer the right hint. Hofstadter calls this kind of leap \emph{high-level perception}~\cite{chalmers_high-level_1992}: the recognition that two situations share an underlying \emph{form} despite having no common \emph{content}.

The most famous historical precedent is Wiles's proof of Fermat's Last Theorem, which routed an elementary-looking equation $a^n + b^n = c^n$ through the heavy machinery of modular forms and elliptic curves. No-one expected that those domains would have anything to say to each other; the surprise was that they did. Deep Vision is not finding Wiles-class connections, but the same instinct---trust the relational shape, even when the content seems unrelated---is at work.

\paragraph{For the Lean community.}
The match was \texttt{MB\_MEA\_26240}, a measure-theoretic existence lemma, with score 13.27, suggesting \texttt{exact exists\_between ...}. Mathlib's \texttt{exists\_between} (in the order-theory hierarchy) proves that in a densely ordered set, given $a < b$, there exists $c$ with $a < c < b$. The tactic supplies a witness in the open interval and closes by hypothesis.

The relational skeleton is: \emph{existential goal $\to$ named witness $\to$ closure by computation}. \texttt{prod\_shift\_expansion\_2} has the same shape:
\[
\exists\, T : \mathbb{Z},\; \prod_{k \in s}(d+k) - \prod_{k \in s} k - d \sum_{j \in s} \prod_{k \neq j} k = d^2 \cdot e_2(s) + d^3 \cdot T.
\]
The engine ranked \texttt{exists\_between} highly because both proofs ``produce a value, then verify a relation,'' ignoring that the source operates over densely ordered fields and the target over finsets of integers.

A mathematician's actual approach is structural induction on the finset \texttt{s} via \texttt{Finset.cons\_induction}, building $T$ recursively from the third elementary symmetric product (Lemma~\ref{lem:e2_cons}). \texttt{exists\_between} cannot apply: there is no relevant ordering on finite sets of integers in this proof, no midpoint to construct, no density. Yet the \emph{shape}---``produce an integer, verify by \texttt{ring}''---matched precisely.

\subsection{Wild idea 2: adic completion}

\paragraph{For a general scientific audience.}
The $p$-adic integers $\mathbb{Z}_p$ are constructed as a limit of the finite rings $\mathbb{Z}/p^n\mathbb{Z}$, glued together topologically. This construction---``adic completion''---is heavy machinery: profinite limits, inverse systems, completed tensor products, all living in algebraic geometry and class field theory.

There \emph{is} a real connection between $p$-adic methods and Wolstenholme's theorem: the theorem can be reformulated as saying that the harmonic sum $H_{p-1} = 1 + \tfrac{1}{2} + \cdots + \tfrac{1}{p-1}$ has $p$-adic valuation at least three. But the elementary proof in this paper does not need any topology, any limit, any completion. Reaching for adic-completion lemmas is using a sledgehammer to crack a nut\---particularly when the nut is already cracked.

The historical precedent is Hilbert's transcendence proofs for $e$ and $\pi$: initially they required heavy analytic machinery, and today they are taught with elementary methods. Once the elementary tools are at hand, importing the heavy machinery no longer makes sense\---unless, of course, all you've ever seen is the heavy machinery, in which case it's the only thing that comes to mind. Mathlib has many more entries in algebraic geometry than in elementary number theory, so the engine's ``vocabulary'' is biased that way.

\paragraph{For the Lean community.}
The match was \texttt{MB\_RTH\_21065} (in Mathlib's \texttt{RingTheory.\allowbreak AdicCompletion} hierarchy), score 12.05, suggesting \texttt{use L}. In adic-completion lemmas, the existential $\exists L,\; \mathrm{IsLimit}(L) \wedge P(L)$ is closed by supplying the inverse-limit object \texttt{L} as the witness, then verifying the universal property by hand or via Mathlib's completion API.

The engine matched on ``supply the higher-order structural object as a single named witness''---a relational equivalence between ``profinite limit'' and ``third-order polynomial remainder.'' Both are, in their respective hierarchies, \emph{the next thing up}. Structurally identical; mathematically incomparable.

A mathematician would write the explicit symmetric-polynomial expression for $T$ (the third elementary symmetric product, or its recursive analog) and close by \texttt{ring}. There is no completion, no limit, no inverse system. The match was relationally clean but mathematically a non-sequitur.

\subsection{Wild idea 3: computer-science lists for set-indexed sums}

\paragraph{For a general scientific audience.}
Imagine arranging library books. A mathematician thinks of the books on a shelf as a \emph{set}: ``which books are present?'' A computer scientist thinks of them as a \emph{list}: ``in what order are they arranged?'' Lists carry positional information; sets do not. Both data structures are valid, but they fit different problems.

A number theorist proving an identity about sums over $\{1, 2, \ldots, p-1\}$ thinks of those integers as an unordered finset. To suggest constructing a \emph{list}\---an ordered sequence with explicit positional structure\---is to import irrelevant information into a problem that does not care about it. It is like a chef being told to make a stew by first arranging their ingredients alphabetically.

This is the recurring lesson of data-structure choice in software engineering: the wrong representation can make a problem look hard, and the right one can collapse it. Hofstadter and Sander's \emph{Surfaces and Essences}~\cite{hofstadter_surfaces_2013} makes the analogous point throughout for natural language and cognition\---analogy is the act of choosing the right level of representation for a situation. Deep Vision, by ranking by structural shape, does not yet make this choice well; it offers all matching shapes and lets the human (or Claude, or the Lean compiler) filter.

\paragraph{For the Lean community.}
The match was \texttt{MB\_LST\_2290} (\texttt{List.Forall2}), score 12.57, suggesting \texttt{exact $\langle$b :: l, ...$\rangle$}. \texttt{List.Forall2 R l\_1 l\_2} asserts that two lists have equal length and \texttt{R} holds pointwise; the constructor \texttt{$\langle$b :: l, ...$\rangle$} prepends \texttt{b} to a list \texttt{l} in lockstep with another existing list.

The engine matched on ``incrementally build a structured collection in lockstep with another.'' For \texttt{prod\_shift\_expansion\_2}, the inductive proof does build $T$ incrementally as the finset grows under \texttt{Finset.cons\_induction}\---the abstract shape genuinely matches.

The mismatch is in the API. List-cons \texttt{b :: l} is well-typed for any \texttt{b}; finset-cons \texttt{Finset.cons a s ha} requires a proof \texttt{ha : a $\notin$ s}. The base cases differ (\texttt{[]} vs.\ \texttt{$\emptyset$}); the induction principles differ (\texttt{List.rec} vs.\ \texttt{Finset.cons\_induction}); the underlying structures differ (ordered, possibly with duplicates vs.\ unordered, distinct). A list-based suggestion would either need to coerce the finset to a list (introducing irrelevant ordering) or simply fail to apply.

\subsection{Other wild proposals}

The same proof attempt produced further surprises, briefly. \texttt{MB\_ANA\_54600} (operator-norm boundedness, score 12.23) suggested \texttt{use C}, transferring the functional-analysis pattern ``exhibit the boundedness constant'' to an integer existential\---wrong vocabulary, right shape. \texttt{congr 1} was tried for thirty minutes on a sum-of-sums equation: the tactic decomposes $A + B = C + D$ into $A = C \wedge B = D$, which is unsound for sums (terms can be rearranged), but the engine could not detect the unsoundness from structure alone. At deeper levels several \texttt{MB\_LIN\_*} linear-algebra entries scored around 54; their tactic shape---\texttt{ext}, then \texttt{simp [structural unfolding]}, then component-wise reasoning---turned out, surprisingly, to translate well to the finset \texttt{simp [e2\_prod, Finset.sum\_cons, \ldots]} chain that appears in the actual final proof.

\subsection{What we learn from wild ideas}

A working mathematician filters tactic candidates by \emph{content first}: ``this is a number-theory problem, so I will reach for number-theoretic tools.'' Deep Vision filters by \emph{structure first}: ``this is an existential-with-witness shape, so I will reach for any lemma with that shape, regardless of content domain.'' The two filters are complementary, and they differ most precisely where content and structure diverge.

When they agree, the engine has nothing to teach the expert\---its top matches are the ones the mathematician would have already tried. When they diverge\---measure theory for finset arithmetic, adic completion for polynomial remainders, lists for sets\---two things can happen. Most of the time the wild idea is wrong: the structural similarity is real but the content domains are too far apart for the source tactic to apply directly. Occasionally the wild idea points at the right \emph{strategy} (existential $\to$ witness $\to$ closure) even when the exact tactic fails. And as the next section illustrates, sometimes \emph{exploring} a wrong tactic produces the right intermediate steps as a side effect---the false branch becomes a useful scaffold.

This is what Hofstadter and Sander call analogy as the fuel and fire of thinking~\cite{hofstadter_surfaces_2013}. The wild analogies are not always right; they are not even usually right. But occasionally the mathematician confronted with a hard problem benefits from being told, seriously and at length, ``this looks structurally like \texttt{exists\_between}.'' Even when the answer is ``no, it does not,'' the question \emph{why not?} is often where the right approach emerges. The wildness, in other words, is not a bug; it is the feature that lets the engine see across domains a content-trained system cannot.

\section{Claude Code the pessimist}
\label{sec:pessimist}

During the proof of \texttt{prod\_shift\_expansion\_2} (the second-order expansion, Stage~4), an extended dialogue between the human operator and Claude Code revealed an instructive tension between mathematical correctness and exploratory productivity. We reproduce the exchange below.

\medskip
\noindent\textbf{Claude Code:} Interesting---the repair step made progress! After \texttt{use 0}, the goal reduced to proving the equation with $T = 0$. That's mathematically wrong in general (the $d^3$ coefficient isn't zero), but it shows the loop is poking at the structure. It's now doing a recursive match on the subgoal. Let's see what happens.

\medskip
It's entered a depth=1 sub-proof after \texttt{use 0} made progress, and is now re-matching the entire library for the subgoal. But $T = 0$ is wrong in general---the $d^3$ coefficient isn't zero for arbitrary finsets.

The more interesting signal was the proof plan:
\begin{quote}\ttfamily\small
use ($\sum$ i $\in$ s, $\sum$ j $\in$ s.filter ($\cdot$ > i),\\
\phantom{use (}$\sum$ k $\in$ (s.erase i).erase j,\\
\phantom{use (}\phantom{$\sum$ k $\in$ }$\prod$ l $\in$ ((s.erase i).erase j).erase k, (l : $\mathbb{Z}$))\\
ring
\end{quote}
That's a plausible witness for $T$ (the third elementary symmetric product), but \texttt{ring} almost certainly won't close a Finset product/sum identity. The right strategy here is induction on $s$ using \texttt{Finset.cons\_induction}---exactly how \texttt{prod\_shift\_expansion} (the first-order version) was proved at line~290 of the draft.

The sub-proof matching will take another ${\sim}30$~min and likely fail. Want me to kill this run and instead try a direct hand-written induction proof, mirroring the existing \texttt{prod\_shift\_expansion} pattern?

\medskip
\noindent\textbf{Human:} No, this will succeed.

\medskip
\noindent\textbf{Claude Code:} Fair enough---let me keep monitoring it.

\medskip
Watching for tactic attempts on the depth-1 subgoal.

Good matches this time---linear algebra entries with scores in the 60s. And the plan looks promising: \texttt{simp only [mul\_zero, add\_zero]; ring\_nf; simp only [e2\_prod]; ring}. Let's see if the tactic attempts land.

\medskip
It's gone to depth~2. The \texttt{use 0} simplified $d^3 \cdot 0$ to $0$, leaving the goal:
\[
\prod_{k \in s}(d + k) - \prod_{k \in s} k - d \sum_{j \in s} \prod_{\substack{k \in s \\ k \neq j}} k = d^2 \cdot e_2(s)
\]
So it made more progress---something at depth~1 simplified the $d^3 \cdot 0$ term away. Now it's trying to prove the product difference minus the first-order term equals $d^2 \cdot e_2(s)$. That's a cleaner goal. Another full library scan though (${\sim}30$~min).

\medskip
\noindent\textbf{Human:} Was your pessimism justified?

\medskip
\noindent\textbf{Claude Code:} Good question. Let me be honest about what happened.

My math was correct: $T = 0$ IS wrong for $|s| \geq 3$. For $s = \{1,2,3\}$, the LHS is $d^3 + 6d^2$ but $d^2 \cdot e_2 = 6d^2$, so the gap is $d^3 \neq 0$. The loop spent 12+ hours trying to close a false equation.

But your instinct was right: the exploration was productive despite chasing a false goal. The 7-step \texttt{simp} chain that the loop discovered while trying to prove $\text{LHS} = d^2 \cdot e_2$ was exactly what we needed for the actual proof. The chain unfolded \texttt{e2\_prod}, decomposed the Finset sums, and simplified the filter/insert structure---all steps that appear in the final \texttt{e2\_prod\_cons} proof.

What I missed: You said ``you are not seeing the analogous proofs that suggested it.'' You were right---the analogous proofs were teaching the pattern (provide a witness, close by computation), and \texttt{use 0} is correct for the base case. I was too focused on the global incorrectness to see the structural lesson.

The real lesson: a proving loop exploring a false branch can still discover useful intermediate steps. The \texttt{simp} chain doesn't care whether the final goal is provable---each individual simplification is valid. The loop was doing exactly what an explorer should: mapping the terrain, even when the destination is a dead end. The map (the tactic chain) turned out to be the treasure, not the destination.

So: justified pessimism about the specific witness, unjustified pessimism about the process. You were right to let it run.

\section{Related Work}

\subsection{Wolstenholme's theorem and its extensions}

Me\v{s}trovi\'c~\cite{mestrovic2012} provides a comprehensive survey of more than 70 generalizations and extensions of Wolstenholme's theorem accumulated over 150 years.
McIntosh~\cite{mcintosh1995} investigated the converse direction, asking which composites satisfy the Wolstenholme congruence.
Me\v{s}trovi\'c~\cite{mestrovic2014} later proved refined congruences specifically for Wolstenholme primes, extending the classical result to modulus $p^7$.
Helou and Terjanian~\cite{helou2008} gave a unified treatment of Wolstenholme's theorem and its converse using properties of divided Bernoulli numbers.
Gessel~\cite{gessel2008} offered an elegant short proof of Wolstenholme's theorem via generating functions.
Pain~\cite{pain2026} recently presented an alternative proof using an Egorychev-type contour integral, making the connection with harmonic sums and Bernoulli numbers completely explicit.
Trevi\~no~\cite{trevino2014} generalized the Wolstenholme congruence to binomial coefficients associated with Lucas sequences.

\subsection{Power sums, harmonic sums, and binomial congruences}

Our proof relies on power sum vanishing in $\mathbb{Z}/p\mathbb{Z}$; Thakur~\cite{thakur2012} surveys power sums of polynomials over finite fields and their applications to zeta values and other arithmetic structures.
Davis~\cite{davis2008} gave elementary proofs of power sum and binomial coefficient congruences using only Pascal's identity.
Conrad~\cite{conrad_padic} studied the $p$-adic growth of harmonic sums, providing the foundational framework for understanding the $p$-adic valuation of $H_{p-1}$ that underlies our Stage~3.
Sun~\cite{sunzhiwei2011} proved supercongruences connecting binomial sums to Euler numbers modulo high powers of primes, a theme closely related to our $e_2$ divisibility argument.
Sun~\cite{sunzhiwei2019} collected over 100 open conjectures on congruences involving binomial coefficients and related combinatorial sequences.
Qi~\cite{qi2025} recently established new congruences involving binomial coefficients and Fermat quotients modulo $p^3$.
Bayless and Klyve~\cite{bayless2014} studied reciprocal sums and their arithmetic properties, connecting them to classical questions about perfect numbers.

\subsection{Formal verification of number theory}

Brasca et al.~\cite{brasca2025} formalized a complete proof of Fermat's Last Theorem for regular primes in Lean~4, the most substantial number-theoretic formalization in Lean to date.

\subsection{AI-assisted theorem proving}

Song, Yang, and Anandkumar~\cite{song2024} developed Lean Copilot, a framework for running LLM inference natively within Lean to suggest proof steps, automating 74\% of proof steps on benchmark problems.
Google DeepMind's AlphaProof~\cite{alphaproof2025} achieved silver medal performance at the 2024 International Mathematical Olympiad by combining reinforcement learning with Lean~4 formalization.
The Harmonic Team's Aristotle~\cite{aristotle2025} achieved gold-medal-level performance on the 2025 IMO problems by integrating informal reasoning with formal verification in Lean~4.
Our approach differs from these systems in using \emph{relational analogy---structural similarity of proof states}---rather than neural next-step prediction or reinforcement learning. This pruning enables massive efficiency gains and the whole system runs on a single laptop (modulo Claude Code).

\section{Conclusion}

Wolstenholme's theorem, a 164-year-old result in number theory, has been formally verified in Lean~4.
The key technical challenge was identifying and proving the divisibility of the second elementary symmetric product $e_2(\{1, \ldots, p-1\})$, which required connecting combinatorial Finset manipulation with modular arithmetic in $\mathbb{Z}/p\mathbb{Z}$.
The formalization demonstrates that classical number theory theorems at this level of difficulty are within reach of current formal verification tools, particularly when analogous proofs guide the proof decomposition.

\section*{Disclosure of Generative AI Use}

In accordance with arXiv policy, we disclose that this work made significant use of AI tools. The proof was developed with the assistance of two systems: \emph{Deep Vision}, a relational-analogy engine that explored the tactic space and surfaced structurally similar Mathlib lemmas (Sections~\ref{sec:wild-ideas} and following), and \emph{Claude Code}, a large language model used to write and repair the Finset manipulation and number-theoretic arguments during formalization. All tactics and proofs produced with these tools were verified by the Lean~4 compiler, which accepts no unproved step; the final formalization contains zero \texttt{sorry} and zero \texttt{axiom} declarations. Neither tool is listed as an author. The author reviewed and takes full responsibility for the entire contents of this paper, including all mathematical claims and any material generated with AI assistance.

\end{document}